\documentclass[epj,referee]{svjour}
\usepackage{graphics}
\usepackage{amsmath}
\usepackage{amssymb}

\begin{document}
\title{Exponential wealth distribution in a random market. \\ A rigorous explanation}
%\subtitle{}
\author{Jos\'e-Luis Lopez\inst{1} \and Ricardo L\'opez-Ruiz\inst{2} \and Xavier Calbet\inst{3}}                     
% Do not remove
%
\institute{Dep. of Mathematical and Informatics Engineering, 
Universidad P\'ublica de Navarra, Pamplona, and BIFI, Spain \and
Dep. of Computer Science \& BIFI, Faculty of Science, 
Universidad de Zaragoza, Zaragoza, Spain \and 
BIFI, Inst. de Biocomput. y F\'{\i}sica de Sistemas Complejos,
Universidad de Zaragoza, Zaragoza, Spain}
\date{}
% The correct dates will be entered by Springer
%
\abstract{
In simulations of some economic gas-like models, the asymptotic regime shows
an exponential wealth distribution, independently of the initial wealth distribution 
given to the system. The appearance of this statistical equilibrium for this 
type of gas-like models is explained in a rigorous analytical way. 
} %end of abstract
\maketitle
\section{Introduction}
\label{intro}

Exponential (Gibbs) distribution is ubiquitous in many natural phenomena.
It is found in many equilibrium situations ranging from physics to economy.
In an economic context, the society of the Western (capitalist) economies 
can be divided into two distinct groups according to the incomes of people \cite{dragu2001}. 
The 95\% of the population, the middle and lower economic classes of society, 
shows an exponential wealth distribution. The rest of the population, the 5\% of individuals, 
arranges their privileged incomes in a power law distribution. 

Gas-like models interpret economic transactions between agents trading with money
similarly to collisions in a gas where particles exchange energy \cite{yakoven2009,chakra2010}. 
Different gas-like models have successfully been proposed to reproduce the 
two different types of statistical behavior before mentioned.
The exponential distribution can be obtained by supposing a gas of agents
that exchange money in binary collisions, or in first-neighbor interactions, 
with random, deterministic or chaotic ingredients in the selection of agents
or in the exchange law \cite{dragu2000,estevez2008,pellicer2011}. The power law 
behavior is obtained by introducing some kind of inhomogeneity in the gas 
through the saving propensity of the agents \cite{chatter2004} or through the breaking 
of the pairing symmetry in the exchange rules \cite{pellicer2011} .

In this work, we take a continuous version of one of such 
models \cite{lopezruiz2011} and  we show that the exponential 
distribution is the only fixed point of the discrete time evolution 
operator associated to this particular model.
This rigorous analytical result helps us to understand 
the computational statistical behavior of such a model
which reaches its asymptotic equilibrium on the exponential distribution 
independently of the out of equilibrium initial state of the system.

\section{The continuous gas-like model}
\label{sec:1}

We take the simplest gas-like model \cite{dragu2000} in which an ensemble of economic agents
trade with money in a random manner. An initial amount of money is given to each agent, for
instance the same to each one. Transactions between pairs of agents are allowed in such a way
that both agents in each interacting pair are randomly chosen. Also they share their money 
in a random way. When the gas evolves with this type of interactions, the exponential distribution
appears as the asymptotic wealth distribution. In this model, the local interactions conserve 
the money, then the global dynamics is also conservative and the total amount of money is 
constant in time. 

The discrete version of this model is as follows \cite{dragu2000}.
The trading rules for each interacting pair $(m_i,m_j)$ of the ensemble of $N$ economic 
agents can be written as
\begin{eqnarray}
m'_i &=& \epsilon \; (m_i + m_j), \nonumber\\
m'_j &=& (1 - \epsilon) (m_i + m_j), \label{model1}\\
i , j &=& 1 \ldots N, \nonumber
\end{eqnarray}
where $\epsilon$ is a random number in the interval $(0,1)$. 
The agents $(i,j)$ are randomly chosen. Their initial money $(m_i, m_j)$, 
at time $t$, is transformed after the interaction in $(m'_i, m'_j)$ at time $t+1$. 
The asymptotic distribution $p_f(m)$, obtained by numerical simulations, 
is the exponential (Boltzmann-Gibbs) distribution,
\begin{equation}
p_f(m)=\beta \exp(-\beta \,m), \hspace{0.5cm}\hbox{with}\hspace{0.5cm}\beta={1/ <m>_{gas}},
\label{eq-exp}
\end{equation}
where $p_f(m) {\mathrm d}m$ denotes the PDF ({\it probability density function}), i.e.
the probability of finding an agent with money (or energy in a gas system) between 
$m$ and $m + {\mathrm d}m$. 
Evidently, this PDF is normalized, $\vert\vert p_f\vert\vert=\int_0^{\infty} p_f(m){\mathrm d}m=1$. 
The mean value of the wealth, $<m>_{gas}$, can be easily calculated directly from the gas
by $<m>_{gas}=\sum_i m_i/N$.   

The continuous version of this model consists in thinking 
that an initial wealth distribution $p_0(m)$, with a mean wealth value $<p_0>$,
evolves from time $n$ to time $n+1$ under the action of an operator $T$ to asymptotically 
reach the equilibrium distribution $p_f(m)$, i.e.
\begin{equation}
\lim_{n\rightarrow\infty} {T}^n \left(p_0(m)\right) \rightarrow p_f(m).
\label{eq-operT}
\end{equation}
In this particular case, $p_f(m)$ is the exponential distribution
with the same average value $<p_f>=<p_0>$, due to the conservation of the total money. 
This is a macroscopic interpretation of Eqs. (\ref{model1}) in the sense that,
under this point of view, each iteration of the operator $T$ means 
that many interactions, of order $N/2$,
have taken place between different pairs of agents. As $n$
indicates the time evolution of $T$, we can roughly assume that $t\approx N*n/2$,
where $t$ follows the microscopic evolution of the individual transactions (or collisions) 
between the agents (this alternative microscopic interpretation can be seen in \cite{calbet2011}).

\begin{figure}[h]
  \resizebox{0.90\columnwidth}{!}{\includegraphics{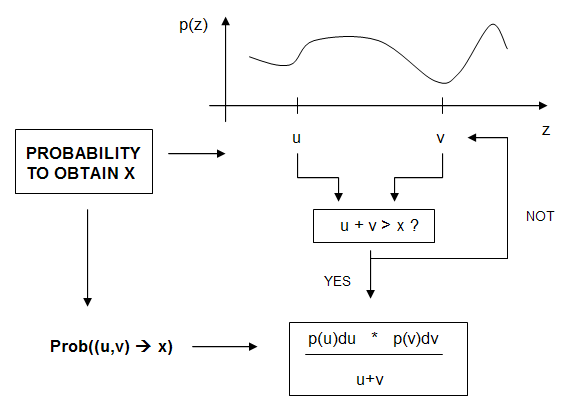}}
  \caption{Derivation of the nonlinear evolution operator $T$ for the 
  agent-based economic model given by Eqs. (\ref{model1}). See explanation in the text.}
  \label{Fig1}
\end{figure}

Now, we recall how to derive  $T$ (see Ref. \cite{lopezruiz2011}).
Suppose that $p_n$ is the wealth distribution in the ensemble at time $n$. 
The probability that two agents with money $(u,v)$ interact will be
$p_n(u)*p_n(v)$. As the trading is totally random, their exchange can give rise with equal
probability to any value $x$ comprised in the interval $(0,u+v)$. Then, the probability to
obtain a particular $x$ (with $x<u+v$) for the interacting pair $(u,v)$ will be 
$p_n(u)*p_n(v)/(u+v)$. Finally, we can obtain the probability to have money $x$ at time $n+1$.
It will be the sum of the probabilities for all the pairs of agents $(u,v)$ able to generate the
quantity $x$, that is, all the pairs verifying $u+v>x$. $T$ has then the form
of a nonlinear integral operator (see Fig. \ref{Fig1}),
\begin{equation}
p_{n+1}(x)={T}p_n(x) = \int\!\!\int_{u+v>x}\,{p_n(u)p_n(v)\over u+v} 
\; {\mathrm d}u{\mathrm d}v \,.
\label{eq-T}
\end{equation}

If we suppose $T$ acting in the PDFs space, we prove in the next section that $T$
conserves the mean wealth of the system, $<Tp>=<p>$. It also conserves 
the norm ($\vert\vert \cdot\vert\vert$), i.e. $T$ maintains the total number of agents 
of the system, $\vert\vert T p\vert\vert=\vert\vert p\vert\vert=1$, that 
by extension implies the conservation of the total richness of the system. 
We also show that the exponential distribution $p_f(x)$ with the right average value
is the only steady state of $T$, i.e. $T p_f=p_f$, and that order two periodic points 
do not exist. It also seems feasible that other high period orbits do not exist.
In consequence, it can be argued that the relation (\ref{eq-operT}) is true. 
See for instance simulations done in Fig. \ref{Fig2}. The same result is obtained for
any other initial distribution \cite{shiva2011}.
 
\begin{figure}[h]
  \centering \resizebox{0.90\columnwidth}{!}{\includegraphics{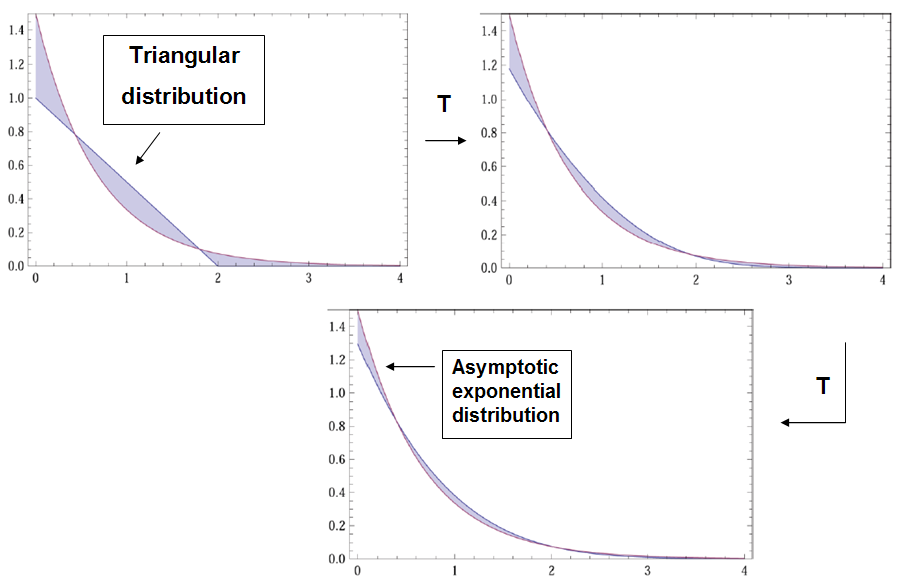}} 
  \caption{An example of the action of $T$ in the space of PDFs. In this case, $p_0(x)$ 
  is chosen as a normalized triangular distribution. Observe that after two iterations of $T$
  the system has practically reached the asymptotic state, i.e. the exponential distribution $p_f(x)$.
  (Figure avanced from \cite{shiva2011}). }
  \label{Fig2}
\end{figure}

\section{Properties of operator $T$}
\label{sec:2}

\subsection{Definitions}

\begin{definition}
We introduce the space $L_1^+$ of positive functions (wealth distributions)
in the interval $[0,\infty)$,
$$
L_1^+[0,\infty)=\lbrace y:[0,\infty)\to R^+ \cup\lbrace0\rbrace,
\hskip 2mm \vert\vert y\vert\vert<\infty\rbrace,
$$
with norm
$$
\vert\vert y\vert\vert=\int_0^\infty y(x) dx.
$$
\end{definition}

\begin{definition}
\label{def-mean1}
We define the mean richness $<x>_y$ associated to a wealth distribution $y\in L_1^+[0,\infty)$ as
the mean value of $x$ for the distribution $y$. In the rest of the paper,
we will represent it by $<y>$. Then,  
$$
<y> \equiv <x>_y = \vert\vert xy(x)\vert\vert=\int_0^\infty xy(x) dx.
$$

\end{definition}

\begin{definition}
\label{def-region}
For $x\ge 0$ and $y\in L_1^+[0,\infty)$ the action of operator $T$ on $y$ is defined by
$$
T(y(x))=\int\int_{S(x)} dudv{y(u)y(v)\over u+v}, 
$$
with $S(x)$ the region of the plane representing the pairs of agents $(u,v)$ which can
generate a richness $x$ after their trading, i.e.
$$
S(x)=\lbrace (u,v), \hskip 2mm u,v>0,\hskip 2mm u+v>x\rbrace.
$$
\end{definition}

After some manipulations, other representations of the integral operator $T$ are possible:
\begin{equation}
(Ty)(x)={1\over 2}\int_x^\infty{dr\over r}\int_{-r}^r dt\,\, y\left({r+t\over 2}\right)y\left({r-t\over 2}\right),
\label{operator}
\end{equation}
$$
(Ty)(x)=x\left[\int_0^1 dt\int_{1-t}^\infty ds+\int_1^\infty dt\int_0^\infty ds\right]{y(xt)y(xs)\over t+s},
$$
$$
(Ty)(x)=x\left[\int_0^\infty dt\int_0^\infty ds-\int_0^1 dt\int_0^{1-t} ds\right]{y(xt)y(xs)\over t+s},
$$
$$
(Ty)(x)=\int_x^\infty dr\int_0^1 du\,\,y(ur)y((1-u)r). \hskip 3cm
$$

\subsection{Theorems}

\begin{theorem}
\label{teorema-norma} 
For any $y\in L_1^+[0,\infty)$ we have that 
$\vert\vert Ty\vert\vert=\vert\vert y\vert\vert^2$. In particular, 
for $y$ being a PDF, i.e. if $\vert\vert y\vert\vert=1$, then $\vert\vert Ty\vert\vert=1$.
(It means that the number of agents in the economic system is conserved in time).
\end{theorem}

\begin{proof}
Take $y\in L_1^+[0,\infty)$. Then (see Fig. \ref{Fig3})
$$
\vert\vert Ty\vert\vert=\int_0^\infty dx\int\int_{S(x)} dudv{y(u)y(v)\over u+v}=\hskip 1.4cm
$$
$$
=\int_0^\infty du\int_0^\infty dv\int_0^{u+v}  dx{y(u)y(v)\over u+v}= 
$$
$$
=\int_0^\infty y(u)du\int_0^\infty y(v)dv=\vert\vert y\vert\vert^2. \hskip 0.5cm \clubsuit
$$
\end{proof}

\begin{corollary} 
It is clear that $(Ty)(x)>0$ for $x>0$ and $Ty\in L_1^+[0,\infty)$. 
Also, it is straightforward to see that $Ty$ is a continuous function, i.e.
$Ty\in\mathbb{C}[0,\infty)$. Therefore, 
from the above Theorem we have that $Ty\in L_1^+[0,\infty)\bigcap\mathbb{C}[0,\infty)$, 
that is $T: L_1^+[0,\infty)\longrightarrow  L_1^+[0,\infty)\bigcap\mathbb{C}[0,\infty)$.
\end{corollary}

\begin{corollary}  
Consider the subset of PDFs in $L_1^+[0,\infty)$, i.e. the unit sphere 
$B:=\lbrace y\in L_1^+[0,\infty)$, $\vert\vert y\vert\vert=1\rbrace$. Observe that
if $y\in B$ then $Ty\in B$. Therefore, we can define the restriction $T_B$ of $T$ 
to the subset $B$, that is $T_B:B\to B$.
\end{corollary}

\begin{figure}[h]
  \centering \resizebox{0.70\columnwidth}{!}{\includegraphics{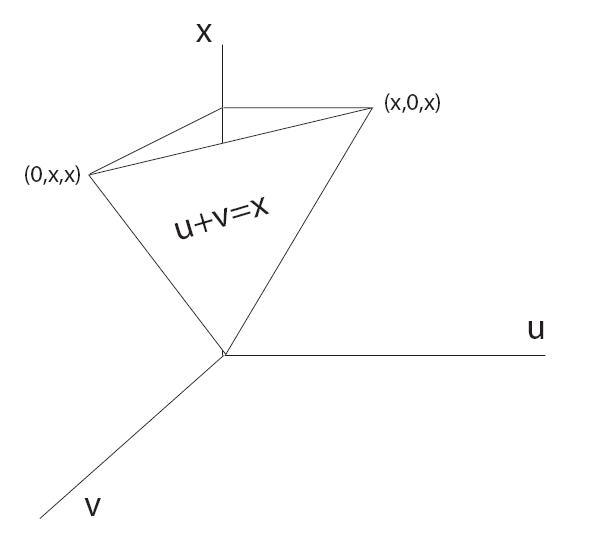}} 
  \caption{Integration region in the triple integral considered in the proof of 
  Theorem \ref{teorema-norma}.}
  \label{Fig3}
\end{figure}

\begin{proposition} 
The operator $T$ is Lipchitz continuous in $B$ with Lipchitz constant $\le 2$.
\end{proposition}
\begin{proof} 
Take $y,w\in B$. Then
$$
\vert\vert Ty-Tw \vert\vert= \hskip 6cm
$$
$$
=\int_0^\infty du\int_0^\infty dv\int_0^{u+v}  dx{\vert y(u)y(v)-w(u)w(v)\vert\over u+v}= \hskip 1cm
$$
$$
=\int_0^\infty du\int_0^\infty dv \vert y(u)y(v)-y(u)w(v)+ \hskip 2cm
$$
$$
\hskip 1cm +y(u)w(v)-w(u)w(v)\vert\le \hskip 1.2cm
$$
$$
\le\int_0^\infty y(u)du\int_0^\infty dv \vert y(v)-w(v)\vert+ \hskip 3cm
$$
$$
+\int_0^\infty w(v)dv\int_0^\infty du  \vert y(u)-w(u)\vert=
2\vert\vert y-w\vert\vert. \hskip 0.5cm \clubsuit
$$
\end{proof}

\begin{theorem}
The mean value $<y>$ of a PDF $y$ is conserved, that is $<Ty>=<y>$ for any $y\in B$.
(It means that the mean wealth, and by extension the total richness, of the economic system
are preserved in time).
\end{theorem}
\begin{proof}
$$
<Ty>=\vert\vert x(Ty)(x)\vert\vert=\int_0^\infty xdx\int\int_{S(x)}dtds{y(t)y(s)\over t+s}=
$$
$$
=\int_0^\infty dt\int_0^\infty ds\int_0^{t+s} xdx{y(t)y(s)\over t+s}=
$$
$$
={1\over 2}\int_0^\infty dt\int_0^\infty ds(t+s)y(t)y(s)=
$$
$$
={1\over 2}\int_0^\infty ty(t)dt+{1\over 2}\int_0^\infty sy(s)ds= 
$$
$$
=\int_0^\infty xy(x)dx= \vert\vert xy(x)\vert\vert=<y>. \hskip 0.5cm \clubsuit
$$
\end{proof}

\noindent
{\bf Observation:}  
Because $\vert\vert Ty\vert\vert=\vert\vert y\vert\vert^2$ in $L_1^+[0,\infty)$, 
this operator divides the space $L_1^+[0,\infty)$ in three regions: the interior of 
the unit ball $B_{in}=\lbrace y\in L_1^+[0,\infty)$, $\vert\vert y\vert\vert<1\rbrace$, 
the unit sphere $B$ defined in Corollary 2, and the exterior of the unit ball 
$B_{ex}=\lbrace y\in L_1^+[0,\infty)$, $\vert\vert y\vert\vert>1\rbrace$. 
Accordingly, the recursion
$$
y_n(x)=  T^ny_0(x), \hskip 0.5cm \hbox{with} \hskip 0.3cm y_0(x) \in  L_1^+[0,\infty),
$$ 
with the norm $\vert\vert\cdot\vert\vert$ has three different behaviors 
depending on $y_0$:
\begin{itemize}
\item If $y_0\in B_{in}$, then $\lim_{n\to\infty}\vert\vert y_n(x)\vert\vert=0$.

\item If $y_0\in B$, then $\vert\vert y_n(x)\vert\vert=1$ 
(the succession $y_n$ remains in $B$ $\forall$ $n$).

\item If $y_0\in B_{ex}$, then $\lim_{n\to\infty}\vert\vert y_n(x)\vert\vert=\infty$.
\end{itemize}

Then, it is obvious that $y=0$ is a fixed point of $T$. The remaining fixed points of $T$ 
in $L_1^+[0,\infty)$, if any, must be located in $B$. In fact, we have the following Theorem.

\begin{theorem}
Apart from $y=0$, the one-parameter family of functions 
$y_\alpha(x)= \alpha e^{-\alpha x}$, $\alpha>0$, are the unique fixed points 
of $T$ in the space $L_1^+[0,\infty)$.
\label{teorema-unicidad}
\end{theorem}

\begin{proof} 
They are fixed points. Using the representation (\ref{operator})
we find (see Figure \ref{Fig4}) that
$$
T\left(\alpha e^{-\alpha x}\right)={\alpha^2\over 2}\int_x^\infty dr\int_{-r}^r dt{e^{-\alpha r}\over r}=
$$
$$
\hskip 2cm=\alpha^2\int_x^\infty e^{-\alpha r}dr=\alpha e^{-\alpha x}.
$$

To see the uniqueness, consider the Fourier transform of $y\in L_1^+[0,\infty)$:
$$
\bar  y(p)=\int_0^\infty e^{i p x}y(x)dx,
$$
and the Fourier transform of $Ty\in L_1^+[0,\infty)$:
$$
(\overline{Ty})(p)=\int_0^\infty e^{i p x}dx\int\int_{S(x)}{dudv\over u+v}y(u)y(v)=
$$
$$
=\int_0^\infty du y(u)\int_0^\infty dv y(v)\int_0^{u+v}dx {e^{i p x}\over u+v}= 
$$
$$
={1\over ip}\int_0^\infty du y(u)\int_0^\infty dv y(v){e^{ip(u+v)}-1\over u+v}.\hskip 0.3cm
$$

If $(Ty)(x)=y(x)$ then $(\overline{Ty})(p)=\bar  y(p)$ and
$$
ip\bar  y(p)=\int_0^\infty du y(u)\int_0^\infty dv y(v){e^{ip(u+v)}-1\over u+v}.
$$
Taking the derivative with respect to $p$ in both sides of this equality we find that 
the Fourier transform $\bar  y(p)$ satisfies the differential equation: 
$$
\bar  y(p)+p\bar  y'(p)=\bar  y^2(p),
$$
whose general solution is
$$
\bar  y(p)={1\over 1-cp}, \hskip 1cm c\in C.
$$
When we are considering only positive functions $y(x)$, their Fourier transform must 
satisfy $\bar y(p)^*=\bar y(-p)$. This is only possible if $c=-c^*$, that is, the constant $c$ 
must be purely imaginary: $c=i/\alpha$, $\alpha\in R$ and then
$$
\bar  y(p)={1\over 1-(i/\alpha)p}, \hskip 1cm \alpha\in R.
$$
From the inverse theorem for the Fourier transform [\cite{cham}, p. 61, Theorem 8.3] we have that 
any solution $y(x)$ of the equation $Ty=y$ in $L_1^+[0,\infty)$ is of the form:
$$
y(x)={1\over 2\pi}\lim_{\epsilon\to 0}\int_{-\infty}^\infty{e^{-ipx-\epsilon
\vert p\vert}\over 1-(i/\alpha)p}dp.
$$
From [\cite{wong}, p. 197, Lemma 1] we have that
$$
y(x)={1\over 2\pi}\int_{-\infty}^\infty{e^{-ipx}\over 1-(i/\alpha)p}dp.
$$
Because $x>0$, the integrand is exponentially small when the imaginary part 
of $p$ is large and negative. Then, we can close the integration line $(-\infty,\infty)$ 
with a semi-infinite circle on the lower complex $p$-plane:
$$
y(x)={1\over 2\pi}\int_{\cal C}{e^{-ipx}\over 1-(i/\alpha)p}dp,
$$
where ${\cal C}$ is the union of the real line and that semicircle traveled in the negative 
sense (clockwise). The pole $p=-i\alpha$ of the integrand is outside of ${\cal C}$ if $\alpha<0$ 
and inside ${\cal C}$ if $\alpha>0$. Applying the Cauchy Residua Theorem to the above integral 
we obtain that 
$$
y(x)=\left\lbrace
\begin{array}{l}
0 \hskip 1.3cm  {\rm if}\hskip 2mm \alpha<0, \\ 
\alpha e^{-\alpha x}\hskip 5mm   {\rm if} \hskip 2mm \alpha>0.
\end{array}
\right.
$$
Then, $y=0$ and $y_\alpha(x)= \alpha e^{-\alpha x}$, $\alpha>0$, are the unique fixed points 
of $T$ in the space $L_1^+[0,\infty)$. $\hskip 0.5cm \clubsuit$
\end{proof}

 \begin{figure}[h]
  \centering \resizebox{0.60\columnwidth}{!}{\includegraphics{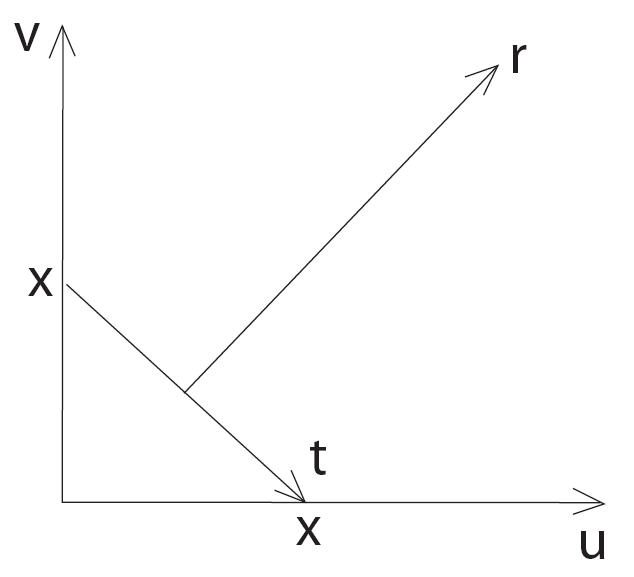}} 
  \caption{Integration region in the double integral considered in the proof 
  of Theorem \ref{teorema-unicidad}.}
  \label{Fig4}
\end{figure}

\begin{corollary} 
The family of functions $y_\alpha(x)= \alpha e^{-\alpha x}$, $\alpha>0$, 
are the unique fixed points of $T_B$ (the restriction of $T$ in the unit ball $B$).
\end{corollary}
%%%

As we will show with some examples in the next section,
the numerical results indicate that, apart from the family
of exponential distributions, there is no other periodic or aperiodic 
asymptotic sets for the operator $T$ acting on the PDFs ball $B$. 
Thus, the whole unit ball $B$ seems to contract
to the exponential steady-state regime. In the next Theorem, we make another
step in this direction by showing that the operator $T$ has no cycles of order two.   

\begin{theorem} 
The operator $T$ has not cycles of order two in $L_1^+[0,\infty)$, that is, 
$T^2y=y\Rightarrow Ty=y$.
\end{theorem}

\begin{proof} 
A similar computation to the one introduced in the proof of Theorem \ref{teorema-unicidad} 
shows that
$$
\left\lbrace\begin{array}{lll}
(\overline{Ty})(p) & =  \displaystyle{\int_0^\infty du y(u)\int_0^\infty dv y(v){e^{ip(u+v)}-1\over ip(u+v)},} \\
(\overline{T^2y})(p) & = \displaystyle{\int_0^\infty du (Ty)(u)\int_0^\infty dv (Ty)(v){e^{ip(u+v)}- 1\over ip(u+v)}}= \\
& = \overline{y}(p). 
\end{array}\right.
$$
Taking the derivative with respect to $p$ in both sides of both equations, we obtain the following  
system of differential equations for $\overline{y}(p)$ and $(\overline{Ty})(p)$:
$$
\left\lbrace\begin{array}{cc}
[p(\overline{Ty})(p)]'= & \overline{y}^2(p), \\ 
\,[p\overline{y}(p)]'= & (\overline{Ty})^2(p). 
\end{array}\right.
$$
Taking the cross product of these equations and multiplying by $p^2$ we get the equation:
$$
p^2\overline{y}^2(p)[p\overline{y}(p)]'=p^2(\overline{Ty})^2(p)[p(\overline{Ty})(p)]' \,,
$$
and integrating:
$$
(\overline{Ty})^3(p)=\overline{y}^3(p)+{c\over p^3}\,, \hskip 1cm c\in \mathbb{C}.
$$
But both, $\overline{y}(p)$, $(\overline{Ty})(p)\in {\cal C}(-\infty,\infty)$ and then $c=0$. 
This means $(\overline{Ty})(p)=\overline{y}(p)$ and $(Ty)(x)=y(x)$ (a.e.).  
$\hskip 0.5cm \clubsuit$
\end{proof}

To finish this section, another interesting property of operator $T$ 
concerning to the value of the norm of the derivatives of $T^ny$ is 
disclosed in the following Proposition. A consequence of this property
is also given.

\begin{proposition} 
For any $y\in L_1^+[0,\infty)$ and $m,n\in N$, with $m\le n$, 
we have the following properties:
$$
(T^ny) \hskip 0.2cm \hbox{is $m$ times differentiable},
$$ 
$$
(-1)^m(T^ny)^{(m)}\in L_1^+[0,\infty),
$$
$$
\vert\vert (T^ny)^{(m)}\vert\vert=\vert (T^ny)^{(m-1)}(0)\vert, 
$$
$$
(-1)^m(T^ny)^{(m)}(0)= \hskip 4cm
$$
$$
\hskip 1.8cm ={1\over m}\sum_{k=0}^{m-1}(T^{n-1}y)^{(k)}(0)(T^{n-1}y)^{(m-k-1)}(0).
$$
\end{proposition}

\begin{proof} 
For fixed $n$, let us show that $(T^ny)$ is $m$ times continuously differentiable and 
$(-1)^m(T^ny)^{(m)}\in L_1^+[0,\infty)$  
by induction over $m$. It is obvious for $m=0$: $(T^ny)$ is continuous and 
$(T^ny)\in L_1^+[0,\infty)$ for any $n=1,2,3,\ldots$ Suppose that it is true for a certain 
$m-1<n$ and let us show that it is also true for $m$. From the last representation of the integral 
operator $T$ given after Definition \ref{def-region}, we have that
$$
(T^ny)'(x)=-\int_0^1 (T^{n-1}y)(ux)(T^{n-1}y)((1-u)x)du<0,
$$
and then
$$
(-1)^m(T^ny)^{(m)}(x)= \hskip 6cm
$$
$$
\sum_{k=0}^{m-1}\dbinom{m-1}{k}\int_0^1 u^k(1-u)^{m-k-1}\left[(-1)^k(T^{n-1}y)^{(k)}(ux)\right]
$$
$$
\hskip 1.7cm \times\left[(-1)^{m-k-1}(T^{n-1}y)^{(m-k-1)}((1-u)x)\right]du. 
$$
From the induction hypothesis, all the integrands in the above sum of integrals are continuous 
and then $(T^ny)^{(m)}$ is continuous and integrable. Moreover, they are positive and 
then $(-1)^m(T^ny)^{(m)}(x)>0$. 
The last two assertions of the Proposition are trivial. $\hskip 0.5cm \clubsuit$
\end{proof}

As a consequence of the former Proposition, we have:

\begin{proposition} 
A fixed point $y$ of $T$, if any, must be a completely monotonic function: 
$y\in{\cal C}^\infty[0,\infty)$ and $(-1)^my^{(m)}(x)\ge 0$ $\forall x\in[0,\infty)$.
Then, as it can be easily tested, the exponential distribution
$y_\alpha(x)= \alpha e^{-\alpha x}$, $\alpha>0$, is one of 
such type of functions. 
\end{proposition}

\begin{proof}
If $Ty=y$ for a certain $y\in L_1^+[0,\infty)$, 
then $T^ny=y$ for any $n\in N$. For arbitrary $m\in N$, take $n\ge m$. 
Our goal is proved just seeing from the former Proposition 
that $(-1)^m y^{(m)}(x)=(-1)^m(T^ny)^{(m)}(x)\ge 0$.
$\hskip 0.5cm \clubsuit$
\end{proof}

\section{Time evolution of $T$: Examples}
\label{sec:3}

As discussed in the former section, the iteration of $T$ for any PDF $y$
in the unit ball $B$ with a finite $<y>$ seems to converge toward the only 
fixed point of the system, i.e. the exponential distribution presenting the 
same mean value of the initial condition $y$. In this respect, we make the 
following conjecture:
 
\noindent
{\bf Conjecture}
{\it For any $y\in B$ with a finite $<y>$, 
$$
\lim_{n\to\infty} (T^ny)(x)=\alpha e^{-\alpha x}
$$ 
with
$$
{1\over \alpha}=<y>=\int_0^\infty xy(x)dx.
$$
}
From simulations (see Fig. \ref{Fig2}), it can be guessed 
that the convergence suggested in this conjecture
is fast. Next, we compute and verify the contraction of the dynamics toward 
the exponential distribution under the action of $T$ for some particular PDFs.
  
\begin{example}
For $\alpha>0$ and $n\in\mathbb{N}$, consider the family 
$$
y_{\alpha,n}(x)={\alpha ^{n+1}\over n!}x^ne^{-\alpha x}\in B.
$$
The mean value is:
$$
<y_{\alpha,n}>=\gamma^{-1}=(n+1)/\alpha.
$$
The limit distribution of $T^ny_{\alpha,n}$ when $n\to\infty$
should be $w(x)=\gamma e^{-\gamma x}$. The first iteration is:
$$
(Ty_{\alpha,n})(x)={\alpha \sqrt\pi\over 2^{2n+1}n!\Gamma(n+3/2)}\Gamma(2n+1,\alpha x).
$$
We numerically check that 
$$
\vert\vert Ty_{\alpha,n}-w\vert\vert<\vert\vert y_{\alpha,n}-w\vert\vert
$$
for different values of $\alpha$ and $n$.
\end{example}

\begin{example}
For $\alpha$, $\beta>0$, consider the new family
$$
y_{\alpha,\beta}(x)={1\over 2}\left[\alpha e^{-\alpha x}+\beta e^{-\beta x}\right]\in B.
$$
The mean value is:
$$
<y_{\alpha,\beta}>=\gamma^{-1}={\alpha\beta\over 2(\alpha+\beta)}.
$$
The limit distribution should be $w(x)=\gamma e^{-\gamma x}$. 
The first iteration of $T$ gives
$$
(Ty_{\alpha,\beta})(x)= \hskip 6cm
$$
$$
={1\over 4}\left\lbrace \alpha e^{-\alpha x}+\beta e^{-\beta x}+
{2\alpha\beta\over\alpha-\beta}\left[\Gamma(0,\beta x)-\Gamma(0,\alpha x)\right]\right\rbrace.
$$
Again, we numerically check that
$$
\vert\vert Ty_{\alpha,\beta}-w\vert\vert<\vert\vert y_{\alpha,\beta}-w\vert\vert.
$$
for several values of $\alpha$ and $\beta$.
\end{example}

\begin{example}
For $\alpha>0$ and $n\in\mathbb{N}$, consider the $\epsilon$-family of functions
$$
y_{\epsilon,\alpha,n}(x)=(1-\epsilon)\alpha e^{-\alpha x}+
\epsilon{\alpha^{n+1}\over n!}x^n e^{-\alpha x}
$$
with $0\le\epsilon\le 1$. It is straightforward to see that $y_{\epsilon,\alpha,n}\in B$. 
The mean value is: 
$$
<y_{\epsilon,\alpha,n}>=\gamma^{-1}={1+\epsilon n\over\alpha}.
$$
The limit function should be $w(x)=\gamma e^{-\gamma x}$. The first iteration is:
$$
Ty_{\epsilon,\alpha,n}(x)=\alpha\left\lbrace 1+\epsilon\left[\epsilon-2+{2(1-\epsilon)
\over(n+1)!}e^{\alpha x}\Gamma(n+1,\alpha x)+\right.\right.
$$
$$
\hskip 2cm \left.\left.+{\epsilon\sqrt{\pi}\over 4^n n!}{e^{\alpha x}
\over 2\Gamma(n+3/2)}\Gamma(2n+1,\alpha x)\right]\right\rbrace e^{-\alpha x}.
$$
Numerically, it is found that
$$
\vert\vert Ty_{\epsilon,\alpha,n}-w\vert\vert<\vert\vert y_{\epsilon,\alpha,n}-w\vert\vert
$$
for different values of $\epsilon$, $n$ and $\alpha$.
\end{example}

\section{Conclusions}

From a macroscopic point of view, markets have an intrinsic random ingredient 
as a consequence of the interaction of an undetermined set of agents at each moment.
Also the unknowledge of the exact quantities of money exchanged by the agents at each 
moment contributes to the randomness of the global system. Even under these circumstances,
the system can reach an equilibrium. Concretely, it has been found in western societies that
the exponential distribution fits well the wealth distribution for the income of 
the $95\%$ of the population.

A statistical model recently introduced \cite{lopezruiz2011}
that takes into account all these idealistic 
random characteristics has been considered in this work. Some of its mathematical properties 
have been carefully worked out and presented here. These exact results help us to understand 
why the exponential distribution is computationally found as the asymptotic steady state 
of this continuous economic model. 
This fact certifies that markets, as well as many other similar 
statistical systems ranging from physics to biology, can be modeled as dissipative dynamical 
systems having the exponential distribution, or any other distribution if it is the case,
as a stable equilibrium, particularly under totally random conditions.

\section*{Acknowledgments}
J.L.L. and R.L.-R. acknowledge financial support from  the Spanish research projects
DGICYT Ref. MTM2010-21037 and DGICYT-FIS2009-13364-C02-01, respectively.

%%%%%%%%%%%%

\end{document}